\documentclass{article}
\usepackage{kr}

\usepackage{times}
\usepackage{soul}
\usepackage{url}
\usepackage[utf8]{inputenc}
\usepackage[small]{caption}
\usepackage{graphicx}
\usepackage{amsmath}
\usepackage{amsthm}
\usepackage{amssymb}
\usepackage{latexsym}

\usepackage{stmaryrd}

\newtheorem{theorem}{Theorem}

\title{A (Simplified) Supreme Being Necessarily Exists, says the
  Computer: Computationally Explored Variants of Gödel's Ontological Argument}
\author{
    Christoph Benzmüller
    \affiliations
    Freie Universität
    Berlin
    \emails
    c.benzmueller@fu-berlin.de
}

\begin{document}

\maketitle

\begin{abstract}
  An approach to universal (meta-)logical reasoning in classical
  higher-order logic is employed to explore and study simplifications
  of Kurt Gödel's modal ontological argument. Some argument premises are
  modified, others are dropped, modal collapse is avoided and validity
  is shown already in weak modal logics \textbf{K} and \textbf{T}. Key
  to the gained simplifications of Gödel's original theory is the
  exploitation of a link to the notions of filter and ultrafilter in topology.

  The paper illustrates how modern knowledge representation and
  reasoning technology for quantified non-classical logics can
  contribute new knowledge to other disciplines. The contributed material is
  also well suited to support teaching of non-trivial logic formalisms in
  classroom.
\end{abstract}

\section{Introduction}
Variants of Kurt Gödel's \shortcite{GoedelNotes}, resp.~Dana Scott's
\shortcite{ScottNotes}, modal ontological argument have previously
been studied and verified on the computer by Benzmüller and
Woltzenlogel-P.~\shortcite{C40,C55} and Benzmüller and Fuenmayor
\shortcite{J52}, and some previously unknown issues were
revealed in these works,\footnote{E.g.,~the theorem prover \textsc{Leo-II}
  detected that Gödel's
  \shortcite{GoedelNotes} variant of the argument is inconsistent; 
  this inconsistency had, unknowingly, been fixed in the variant of Scott
  \shortcite{ScottNotes}; cf.~\cite{C55} for more details.}
and it was shown that logic \textbf{KB}, instead of  \textbf{S5}, is
already sufficient to derive from Gödel's axioms
that a supreme being necessarily exists.

In this paper simplified variants of Gödel's modal
ontological argument are explored and studied.  Our simplifications have been
developed in interaction with the proof assistant Isabelle/HOL
\cite{Isabelle} and by employing Benzmüller's \shortcite{J41,J23} shallow
semantical embedding (SSE) approach as enabling technology. This
technology supports the reuse of automated theorem proving (ATP) and
model finding tools
for classical higher-order logic (HOL) to represent and reason with a wide
range of non-classical logics and theories, including higher-order
modal logics (HOMLs) and Gödel's modal ontological argument, which are
in the focus of this paper.

One of the new, simplified modal arguments is as follows. 
The notion
of being Godlike ($\mathcal{G}$) is exactly as in Gödel's
original work. 
Thus, a Godlike entity, by definition,  possesses all positive properties
% Remember that
% Gödel defined a Godlike entity to possess all positive properties; this definition remains unchanged
($\mathcal{P}$ is an uninterpreted constant denoting positive properties):
%\[ \includegraphics[width=.5\columnwidth]{Sel1.png} \]
$$\mathcal{G}\ x  \quad \equiv  \quad \boldsymbol{\forall} Y. (\mathcal{P}\ Y
\boldsymbol{\rightarrow} Y\ x)$$

%By definition a Godlike entity thus possesses all positive properties.

The three only axioms of the new theory which constrain the
interpretation of Gödel's positive properties ($\mathcal{P}$) are:
\begin{description}
%\item[\textsf{\small G}]  A Godlike entity to possess all positive properties.
\item[\textsf{\small A1'}]  
Self-difference is not a positive property.\footnote{An alternative to \textsf{\scriptsize  A1'}
  would be: The empty property
  ($\lambda x. \boldsymbol \bot$) is not a positive property.}
% $$\mathcal{P} (\lambda x. x\boldsymbol{=}
%   x) \boldsymbol{\wedge} \boldsymbol{\neg} \mathcal{P}  (\lambda x. x\boldsymbol{\not=}
%   x) $$
$$\boldsymbol{\neg} \mathcal{P}  (\lambda x. x\boldsymbol{\not=}
   x) $$
\item[\textsf{\small A2'}]  
   A property entailed or necessarily entailed by a positive property
   is positive. 
$$\boldsymbol{\forall} X\ Y.\ ((\mathcal{P}\
  X \boldsymbol{\wedge} (X \boldsymbol{\sqsubseteq} Y \boldsymbol{\vee} X
  \boldsymbol{\Rrightarrow} Y)) \boldsymbol{\rightarrow} \mathcal{P}\ Y)$$
\item[\textsf{\small A3}]
%cf.~the Appendices  A--C.\label{footnoteA3} 
  The conjunction of any collection of
  positive properties is positive.\footnote{The
    third-order formalization of \textsf{\scriptsize A3} as given here
  ($\mathcal{Z}$ is a third-order variable ranging over sets of properties) has
  been proposed by Anderson and Gettings \shortcite{AndersonGettings},
  see also Fitting \shortcite{fitting02:_types_tableaus_god}.  Axiom
  \textsf{\scriptsize A3}, together with the definition of
  $\mathcal{G}$, implies that being Godlike is a positive
  property. Since supporting this inference is the only role this
  axiom plays in the argument, $\mathcal{(P\,G)}$ can
  be taken (and has been taken; cf.~Scott \shortcite{ScottNotes}) as an
  alternative to our \textsf{\scriptsize A3}; cf.~also Fig.~\ref{fig:SimpleVariantPG}.} 
$$\boldsymbol{\forall} \mathcal{Z}. (\mathcal{P}os\ \mathcal{Z}
 \boldsymbol{\rightarrow} \boldsymbol{\forall} X. (X \boldsymbol{\textstyle\bigsqcap} \mathcal{Z}
 \boldsymbol{\rightarrow} \mathcal{P}\ X))$$
Technical reading:  if
  $\mathcal{Z}$ is any set of positive properties, then the property
  $X$ obtained by taking the conjunction of the properties in $\mathcal{Z}$ is positive.
\end{description}

%\[ \includegraphics[width=.99\columnwidth]{Sel2.png} \]
In these premises the following defined symbols are used, where $\boldsymbol{\forall}$ is
a possibilist second-order quantifier and where $\boldsymbol{\forall}^\text{E}$ is an
actualist first-order quantifier for individuals:
\begin{align*}
X \boldsymbol{\sqsubseteq} Y & \quad \equiv  \quad \boldsymbol{\forall}^\text{E} z.(X\ z
  \boldsymbol{\rightarrow} Y\ z)\\
X \boldsymbol{\Rrightarrow} Y &  \quad \equiv  \quad \boldsymbol{\Box}(X
                                \boldsymbol{\sqsubseteq} Y)\\
\mathcal{P}os\ \mathcal{Z} &  \quad \equiv  \quad \boldsymbol{\forall} X.
                             (\mathcal{Z}\ X \boldsymbol{\rightarrow}
                             \mathcal{P}\ X)\\
X \boldsymbol{\textstyle\bigsqcap}  \mathcal{Z}  &  \quad \equiv  \quad
                              \boldsymbol{\Box}\boldsymbol{\forall}^\text{E}
                              u. (X\ u \boldsymbol{\leftrightarrow}
                              (\boldsymbol{\forall} Y.\ \mathcal{Z}\ Y
                              \boldsymbol{\rightarrow}  Y\ u)) 
\end{align*}
%\[ \includegraphics[width=.92\columnwidth]{Sel3.png} \]

%Informally we have:

From \textsf{\small A1'},  \textsf{\small A2'} and \textsf{\small A3}
it follows, in a few argumentation steps in modal logic \textbf{K},  that a
Godlike entity  possibly and 
necessarily exists. Modal collapse, which expresses that
there are no contingent truths and which thus eliminates the
possibility of alternative possible worlds, does
not follow from these axioms. Also monotheism is not implied. These observations should render the new theory
interesting to theoretical philosophy and theology.

Compare the above with Gödel's premises of the modal ontological
argument (in the consistent variant of Scott):

\begin{description}
%\item[\textsf{\small G}]  A Godlike entity to possess all positive properties.
\item[\textsf{\small A1}]  
One of a property or its complement is
positive.
$$\boldsymbol{\forall} X. (\boldsymbol{\neg}(\mathcal{P}\ X)
  \boldsymbol{\leftrightarrow} \mathcal{P}
  ({\boldsymbol{\rightharpoondown}}X))$$
\item[\textsf{\small A2}] 
A property necessarily 
entailed by a positive property is positive.
 $$\boldsymbol{\forall} X\ Y. ((\mathcal{P}\
  X \boldsymbol{\wedge} (X
  \boldsymbol{\Rrightarrow} Y)) \boldsymbol{\rightarrow} \mathcal{P}\
  Y)$$
\item[\textsf{\small A3}] The conjunction of any collection of positive properties
is positive (or, alternatively, being Godlike is a positive property).
 $$\boldsymbol{\forall} \mathcal{Z}. (\mathcal{P}os\ \mathcal{Z}
 \boldsymbol{\rightarrow} \boldsymbol{\forall} X. (X \boldsymbol{\textstyle\bigsqcap} \mathcal{Z}
 \boldsymbol{\rightarrow} \mathcal{P}\ X))$$

\item[\textsf{\small A4}] 
Any positive property is necessarily positive.
$$\boldsymbol{\forall} X. (\mathcal{P}\ X 
  \boldsymbol{\rightarrow} \boldsymbol{\Box}(\mathcal{P}\ X))$$ 
\item[\textsf{\small A5}] 
Necessary existence ($\mathcal{NE}$) is a
positive property.
$$\mathcal{P}\ \mathcal{NE}$$ 
\end{description}

Furthermore, axiom \textsf{\small B} is added to ensure that we
operate in logic
\textbf{KB} instead of just 
\textbf{K}.\footnote{Symmetry of the accessibility
relation \textbf{r} associated with modal $\Box$-operator can be postulated alternatively 
in our meta-logical framework.} (Remember that logic \textbf{S5} is not needed.)
$$\boldsymbol\forall \varphi. (\varphi
  \boldsymbol\rightarrow \boldsymbol\Box\boldsymbol\Diamond \varphi)$$ 
% $\forall x\ y.\neg  x \text{r} y\vee y \text{r} x$
 % \footnote{\textsf{\scriptsize B}'s  counterpart $\varphi
  % \boldsymbol\shortrightarrow \Box\Diamond \varphi$
  % is implied and could be used instead.}

%\[ \includegraphics[width=.9\columnwidth]{Sel4.png} \]
Essence ($\mathcal{E}$) and necessary existence ($\mathcal{NE}$)
are defined as (all other definitions are as before):
\begin{align*}
\mathcal{E}\ Y\ x & \quad \equiv  \quad Y\ x \boldsymbol{\wedge} \boldsymbol{\forall} Z.
                  (Z\ x \boldsymbol{\rightarrow} (Y
                    \boldsymbol{\Rrightarrow} Z)) \\
\mathcal{NE}\ x & \quad \equiv  \quad \boldsymbol{\forall} Y.
                  (\mathcal{E}\ Y\ x \boldsymbol{\rightarrow}
                  \boldsymbol{\Box} \boldsymbol{\exists}^\text{E}\ Y)
\end{align*}
%\[ \includegraphics[width=.8\columnwidth]{Sel5.png} \]
Informally: Property $Y$ is an essence $\mathcal{E}$ of an entity $x$ if, and
only if,  (i) $Y$ holds for $x$ and (ii) $Y$ necessarily entails every
property $Z$ of
$x$. Moreover, an entity $x$ has the property of necessary existence
if, and
only if, the essence of $x$ is necessarily instantiated.

% We also give informal readings of Gödel's axioms:
% \textsf{\small A1} says that one of a property or its complement is
% positive. \textsf{\small A2} states that a property necessarily 
% entailed by a positive property is positive, and \textsf{\small A3} is as before.
% \textsf{\small A4} expresses that any positive property is necessarily
% so.
% % (and any negative poperty is necessarily so as well)
% \textsf{\small A5} postulates that necessary existence is a
% positive property. Axiom \textsf{\small B} (symmetry of the accessibility
% relation associated with modal $\Box$-operator) is added  to ensure that we are in modal logic
% \textbf{KB} instead of just 
% \textbf{K}.\footnote{\textsf{\scriptsize B}'s counterpart $\varphi
%   \boldsymbol\shortrightarrow \Box\Diamond \varphi$
%   is implied and could be used instead.}

Using Gödel's premises as stated it can be proved automatically that a
Godlike entity possibly and necessarily exists \cite{C55}. However, modal collapse is still implied even in
the weak logic \textbf{KB}.\footnote{For more information on modal collapse (in
logic \textbf{S5}) consult Sobel \shortcite{Sobel,sobel2004logic}, Fitting
\shortcite{fitting02:_types_tableaus_god} and
Kova\v{c}~\shortcite{Kovacs2012}; see also the
references therein.}

Benzmüller and Fuenmayor \shortcite{J52} recently showed that
different modal ultrafilter properties can be deduced from Gödel's
premises. These insights are key to the argument simplifications
developed and studied in this paper: If Gödel's premises entail that
positive properties form a modal ultrafilter, then why not turning
things around, and start out with an axiom \textsf{\small U1} postulating
ultrafilter properties for $\mathcal{P}$? Then use \textsf{\small U1}
instead of other axioms for proving that a Godlike entity necessarily
exists, and on the fly explore what further simplifications of the
argument are triggered.  This research plan worked out and it led to
simplified argument variants as presented above and in the remainder.
% ,
% where \textsf{\small U1} has been replaced by \textsf{\small
%   A1'} and where \textsf{\small A2} has been strengthened into
%  \textsf{\small A2'} accordingly.

The proof assistant Isabelle/HOL and its integrated ATP
systems have supported our exploration work surprisingly well,
despite the undecidability and high complexity of the underlying logic
setting. 
As usual, we here only present the main steps of the 
exploration process, and various interesting eureka or frustration
steps in between are dropped.

The structure of this paper is as follows: An SSE of HOML in HOL is
introduced in Sect.~\ref{sec:HOML}. This section, parts of which have
been adapted from Kirchner et al.~\shortcite{J47}, ensures that the
paper is sufficiently self-contained; readers familiar with the SSE
approach may simply skip it. Modal filter and ultrafilter are defined
in Sect.~\ref{sec:ultrafilter}. Section~\ref{sec:goedelargument}
recaps, in some more detail, the Gödel/Scott variant of the modal ontological argument from
above. Subsequently, an ultrafilter-based modal ontological argument is
presented in Sect.~\ref{sec:newargument}. This new argument is further
simplified in Sect.~\ref{sec:finalargument}, leading to our proposal
based on axioms \textsf{\small A1'}, \textsf{\small A2'} and
\textsf{\small A3} as presented before.  Further simplifications and
modifications are studied in Sect.~\ref{sec:further}, and related work
is discussed in Sect.~\ref{sec:relwork}.
%Further simplified variants are presented in the Appendix.

Since we develop, explain and discuss our formal
encodings directly in Isabelle/HOL, some familiarity
with this proof assistant and its underlying logic HOL 
%\cite{J43}
\cite{andrews2002introduction,J43} is assumed. The entire sources\footnote{See \href{https://github.com/cbenzmueller/LogiKEy/tree/master/Computational-Metaphysics/2020-KR}{[here]}, or \href{http://logikey.org}{logikey.org} $\rightarrow$
  Computational-Metaphysics. The experiments reported in this paper were conducted on
a standard notebook (2,5 GHz Intel Core i7, 16 GB memory).} of our
formal encodings are presented and explained in 
detail in this paper.

The contributions of this paper are thus manifold. Besides the novel
variants of the modal ontological argument that we contribute to
metaphysics and theology, we demonstrate how the SSE
technique, in combination with higher-order reasoning tools, can
be employed in practical studies to
explore such new knowledge. Moreover, we
contribute useful source encodings that can be reused and adapted to
teach quantified modal logics in interdisciplinary lecture
courses. % Finally, we hint at some possible improvements of the
% reasoning tools we employ.

\section{Modeling HOML in HOL} \label{sec:HOML}
Various SSEs of quantified non-classical logics in HOL have been
developed, studied and applied in related work,
cf.~Benzmüller \shortcite{J41} and Kirchner et al.~\shortcite{J47} and
the references therein.  These contributions, among
others, show
that the standard translation from propositional modal logic to
first-order (FO) logic can be concisely modeled (i.e., embedded) within
HOL theorem provers, so that the modal operator $\Box$, for
example, can be explicitly defined by the $\lambda$-term
$\lambda \varphi. \lambda w. \forall v. (R w v \rightarrow \varphi
v)$,
where $R$ denotes the accessibility relation associated with $\Box$.
Then one can construct FO formulas involving $\Box\varphi$
and use them to represent and proof theorems. Thus, in an SSE, the
target logic is internally represented using higher-order (HO) constructs
in a theorem proving system such as Isabelle/HOL.  Benzmüller and
Paulson \shortcite{J23} developed an SSE that captures quantified
extensions of modal logic. For
example, if $\forall x. \phi x$ is shorthand in HOL for
$\Pi (\lambda x. \phi x)$, then $\Box \forall x Px$ would be
represented as $\Box \Pi' (\lambda x. \lambda w. P x w)$, where $\Pi'$
stands for the $\lambda$-term
$\lambda \Phi . \lambda w . \Pi(\lambda x . \Phi x w)$, and the $\Box$
gets resolved as above.

  To see how
  these expressions can be resolved to produce the right
  representation, consider the following series of reductions:
\[  \begin{tabular}{lll} & & $\Box\forall x P x$ \\ & $\equiv$ &
    $\Box \Pi' (\lambda x. \lambda w. P x w)$\\ & $\equiv$ &
    $\Box ((\lambda \Phi . \lambda w . \Pi(\lambda x . \Phi x w))
    (\lambda x. \lambda w. P x w))$\\
    & $\equiv$ &
    $\Box (\lambda w . \Pi(\lambda x . (\lambda x. \lambda w. P x w) x
    w))$\\
    & $\equiv$ & $\Box (\lambda w . \Pi(\lambda x . P x w))$\\ &
    $\equiv$ &
    $(\lambda \varphi. \lambda w. \forall v. (R w v \rightarrow
    \varphi v)) (\lambda w . \Pi(\lambda x . P x w))$\\
    & $\equiv$ &
    $(\lambda \varphi. \lambda w. \Pi (\lambda v . R w v \rightarrow
    \varphi v)) (\lambda w . \Pi(\lambda x . P x w))$\\
    & $\equiv$ &
    $(\lambda w. \Pi (\lambda v . R w v \rightarrow (\lambda w
    . \Pi(\lambda x . P x w)) v)) $\\
    & $\equiv$ &
    $(\lambda w. \Pi (\lambda v . R w v \rightarrow \Pi(\lambda x . P
    x v)) ) $\\
    & $\equiv$ &
    $(\lambda w. \forall v . R w v \rightarrow \forall x . P x v) $\\
    & $\equiv$ & $(\lambda w. \forall v x . R w v \rightarrow P x v) $
  \end{tabular} 
\]
Thus, we end up with a representation of  $\Box\forall x P x$ in
HOL.  Of course, types are assigned to each (sub-)term of the HOL
language. We assign individual terms (such as variable $x$
above) the type \textsf{\small e}, and terms denoting worlds (such as
variable $w$ above) the type \textsf{\small i}. From such base
choices, all other types in the above presentation can actually be
inferred.

% While types have been
% omitted above, they will  often be given in the remainder.

An explicit encoding of HOML in Isabelle/HOL, following the above
ideas, is presented in Fig.~\ref{fig:HOML}.\footnote{In Isabelle/HOL
  explicit type information can often be omitted due the system's
  internal type inference mechanism. This feature is exploited in our
  formalization to improve readability. However, for \emph{new}
  abbreviations and definitions, we often explicitly declare the
  types of the freshly introduced symbols. This supports a better
  intuitive understanding, and it also reduces the number of
  polymorphic terms in the formalization (heavy use of polymorphism
  may generally lead to decreased proof automation performance).}  In
lines 4--5 the base types \textsf{\small i} and \textsf{\small e} are
declared (text passages in red are comments).  Note that HOL comes with an inbuilt base type
\textsf{\small bool}, the bivalent type of Booleans. No
cardinality constraints are associated with types \textsf{\small i}
and \textsf{\small e}, except that they must be non-empty.  To keep
our presentation concise, useful type synonyms are introduced in lines
6--9.  $\sigma$ abbreviates the type
$\textsf{\small i}{\Rightarrow}\textsf{\small bool}$ ($\Rightarrow$ is the function type constructor
in HOL), and terms of type
$\sigma$ can be seen to represent world-lifted propositions, i.e.,
truth-sets in Kripke's modal relational
semantics~\cite{sep-logic-modal}. The explicit transition from modal
propositions to terms (truth-sets) of type $\sigma$ is a key aspect of the
SSE technique, and in the remainder of this article we use phrases
such as ``world-lifted'' or ``$\sigma$-type'' terms to emphasize this
conversion in the SSE approach. $\gamma$, which stands for
$\textsf{\small e}{\Rightarrow}\sigma$, is the type of world-lifted,
intensional properties. $\mu$ and $\nu$, which abbreviate $\sigma{\Rightarrow} \sigma$
and $\sigma{\Rightarrow} \sigma{\Rightarrow} \sigma$, are the types
associated with unary and binary modal logic connectives.

\begin{figure}[!tb] 
\includegraphics[width=1\columnwidth]{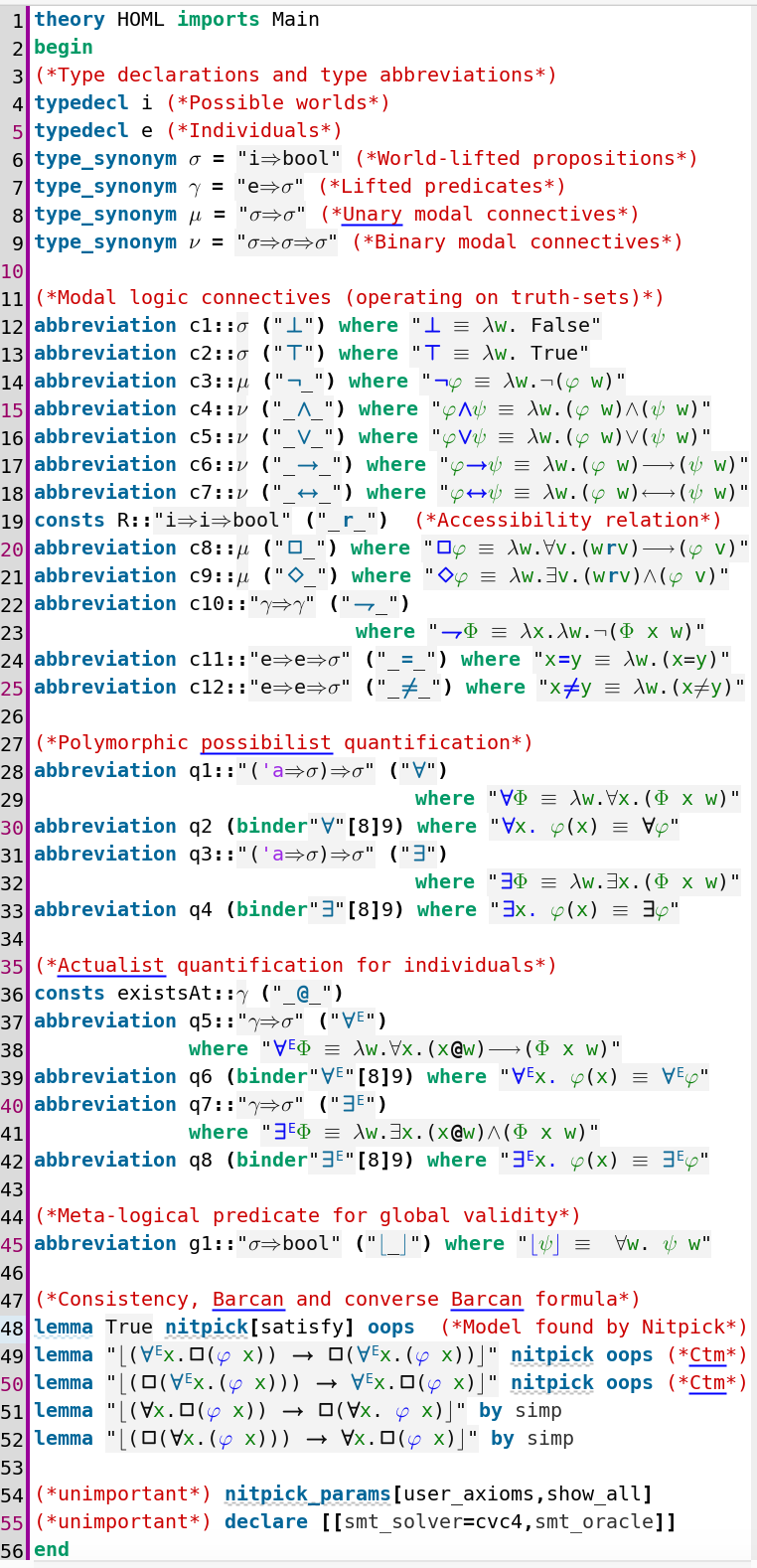}
\caption{SSE of HOML in HOL. \label{fig:HOML}}
\end{figure}

% $\sigma$ (world-lifted propositions/truth-sets; \textsf{\small bool} is the
% HOL type for Booleans and $\Rightarrow$ is the function type
% constructor in HOL), $\gamma$
% (world-lifted predicates), $\mu$ (unary modal connectives), and $\nu$
% (binary modal connectives)

% $\sigma$ abbreviates the type 
% $\textsf{\small i}{\Rightarrow}\textsf{\small bool}$ and termsof type $\sigma$  can be seen to represent world-lifted propositions,
% i.e., truth-sets in Kripke's modal relational
% semantics~\cite{sep-logic-modal}. 

% $\tau$, which abbreviates type
% $\textsf{\small i}{\Rightarrow} \textsf{\small i}{\Rightarrow}\textsf{\small bool}$,
% is accordingly the type of accessibility relations in modal relational
% semantics, and $\gamma$, which stands for $\textsf{\small e}{\Rightarrow}\sigma$,
% is the type of world-lifted, intensional properties.

The modal logic connectives are defined in lines 12--25.
In line 16, for example, the definition of the world-lifted
$\boldsymbol\vee$-connective of type $\nu$ is given; explicit type information
is presented after the
\textsf{\small ::}-token for `\textsf{\small c5}', which is the
ASCII-denominator for the (right-associative) infix-operator $\boldsymbol\vee$ as introduced in
parenthesis shortly after. $\varphi_\sigma \boldsymbol\vee
\psi_\sigma$ is then defined
as abbreviation for the truth-set $\lambda w.  (\varphi_\sigma
w) \vee (\psi_\sigma w)$, respectively. In the remainder we generally use bold-face
symbols for world-lifted connectives (such as $\boldsymbol\vee$) in
order to rigorously distinguish them from their ordinary
counterparts (such as $\vee$) in meta-logic HOL.

Further modal logic connectives, such as $\boldsymbol\bot$,
$\boldsymbol\top$, $\boldsymbol\neg$,
$\boldsymbol\rightarrow$, and $\boldsymbol\leftrightarrow$
are introduced analogously.  The
operator $\boldsymbol\rightharpoondown$, introduced in lines 22--23, is
inverting properties of types $\gamma$; this operation occurs in 
Gödel's axiom \textsf{\small A1}.
$\boldsymbol =$ and $\boldsymbol {\not=}$ are defined in lines 24--25 as
world-independent, syntactical equality.

The world-lifted modal $\boldsymbol\Box$-operator is introduced 
 in lines 19--20; accessibility relation $R$ is now synonymously named
$\textbf{r}$ in infix notation. The  definition of
$\boldsymbol{\Diamond}$ is analogous.

The world-lifted (polymorphic) possibilist quantifier
$\boldsymbol\forall$ as discussed before is introduced in
line 28--29. In line 30, user-friendly binder-notation for
$\boldsymbol\forall$ is additionally defined. Instead of
distinguishing between $\forall$ and $\Pi'$ as in our illustrating
example, $\boldsymbol\forall$-symbols are overloaded here. The introduction of the
possibilist $\boldsymbol\exists$-quantifier is analogous.

Additional actualist quantifiers,
$\boldsymbol\forall^E$ and $\boldsymbol\exists^E$, are introduced in
lines 36--42. Their definition is guarded by an explicit, possibly
empty, \textsf{\small existsAt} (\textsf{\small @}) predicate, which encodes whether an
individual object actually ``exists'' at a particular world, or
not. The actualist quantifiers are declared
non-polymorphic, and they support quantification over individuals
only. In the remainder we will indeed apply
$\boldsymbol\forall$ and $\boldsymbol\exists$ for different types in
the type hierarchy of HOL, while $\boldsymbol\forall^E$ and
$\boldsymbol\exists^E$ exclusively quantify over individuals only.
 
Global validity of a world-lifted formula $\psi_\sigma$,
denoted as $\lfloor\psi\rfloor$, is introduced in line 45 as an
abbreviation for $\forall w_\textsf{\scriptsize i}.  \psi w$.

Consistency of the introduced concepts is confirmed by the model
finder \textit{nitpick} \cite{Nitpick} in line 48. Since only
abbreviations and no axioms have been introduced so far, the
consistency of the Isabelle/HOL theory \textsf{\small HOML} as
displayed in Fig.~\ref{fig:HOML} is actually evident.

In line 49--52 it is studied whether instances of the Barcan and the
converse Barcan formulas are implied.  As expected, both principles
are valid only for possibilist quantification, while they have
countermodels for actualist quantification.

Lines 54--55 declare some specific parameter settings for 
some of the reasoning tools that we employ.

%The following holds.

\begin{theorem}
The SSE of HOML in HOL is faithful (for \textbf{K}).
\end{theorem}
\begin{proof}
Analagous to~Benzmüller and Paulson~\shortcite{J23}.
\end{proof}

Theory \textsf{\small HOML} thus models base logic \textbf{K} in
HOL. Axiom \textsf{\small B}, see above, can be postulated to
arrive at logic \textbf{KB}.

%\newpage

\section{Modal Filter and Ultrafilter} \label{sec:ultrafilter} Theory
\textsf{\small MFilter}, for ``modal filter'', see Fig.~\ref{fig:Ultrafilters},
imports theory \textsf{\small HOML} and adapts the topological notions
of filter and ultrafilter to our modal logic setting. For an 
introduction to the notions of filter and ultrafilter see the literature,
e.g.,~\cite{BurrisSankappanavar1981} or also \cite{odifreddi00:_ultraf_dictat_gods}.
\begin{figure}[tp] \centering
\includegraphics[width=1\columnwidth]{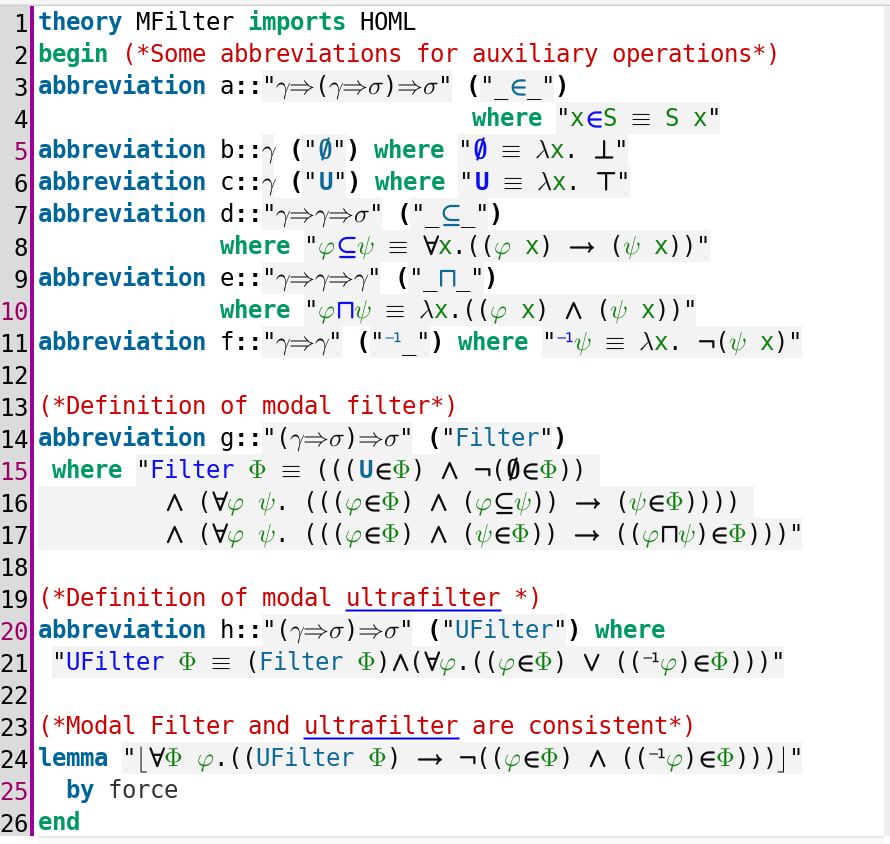}
\caption{Definition of filter and ultrafilter (for possible worlds). \label{fig:Ultrafilters}}
\end{figure}

Our notion of modal ultrafilter is introduced in lines 20--21 as a world-lifted
characteristic function of type $(\gamma{\Rightarrow}\sigma){\Rightarrow}\sigma$. A modal ultrafilter is
thus a world-dependent set of intensions of $\gamma$-type properties; in other 
%$\sigma$-sets of $\sigma$-sets of objects of type $\gamma$.  
words, a $\sigma$-subset of the $\sigma$-powerset of $\gamma$-type
property extensions.  A modal ultrafilter $\phi$ is defined as a modal
filter
satisfying an additional maximality condition:
$\forall \varphi.  \varphi\boldsymbol\in\phi \boldsymbol\vee
(^{-1}\varphi)\boldsymbol\in\phi$,
where $\boldsymbol\in$ is elementhood of $\gamma$-type objects in
$\sigma$-sets of $\gamma$-type objects (see lines 3--4), and where $^{-1}$
is the relative set complement operation on sets of entities (line
11).

A modal filter $\phi$, see lines 14--17, is required to
\begin{enumerate} 
\item be large: $\textbf{U}\boldsymbol\in\phi$, where $\textbf{U}$ denotes
the full set of $\gamma$-type objects we start with, 
\item exclude the empty set:
  $\boldsymbol\emptyset\boldsymbol\not\boldsymbol\in\phi$,
  where $\boldsymbol\emptyset$ is the world-lifted empty set of
  $\gamma$-type objects,
\item be closed under supersets:
  $\boldsymbol\forall \varphi\,
  \psi. (\varphi\boldsymbol\in\phi \boldsymbol\wedge
  \varphi\boldsymbol\subseteq\psi ) \boldsymbol\rightarrow
  \psi\boldsymbol\in\phi$
  (the world-lifted $\boldsymbol\subseteq$-relation  is
  defined in lines 7--8), and
\item be closed under intersections:
  $\boldsymbol\forall \varphi\,
  \psi. (\varphi\boldsymbol\in\phi \boldsymbol\wedge
  \psi\boldsymbol\in\phi) \boldsymbol\rightarrow
  (\varphi\boldsymbol\sqcap\psi)\boldsymbol\subseteq\phi$
  (where $\boldsymbol\sqcap$
  is defined in lines 9--10).
\end{enumerate}

Benzmüller and Fuenmayor \shortcite{J52} have studied
two different notions of modal ultrafilter (termed $\gamma$- and
$\delta$-ultrafilter), which are defined on intensions and extensions of
properties, respectively. This distinction is not
needed in this paper; what we call modal ultrafilter here corresponds to
their notion of $\gamma$-ultrafilter.

\section{Gödel/Scott Variant} \label{sec:goedelargument} We start out in
Fig.~\ref{fig:BaseDefs} 
with the
introduction of some basic abbreviations and definitions for the
Gödel/Scott variant of the modal ontological argument. This theory file,
which is termed \textsf{\small BaseDefs} and which imports \textsf{\small HOML}, is reused without
modification also in all other variants as explored in this paper later on.
% This theory defines some
% basic abbreviations to be reused in all argument variants as 
% explored in the remainder if this paper. 
  \begin{figure}[tp!] \centering
\includegraphics[width=1\columnwidth]{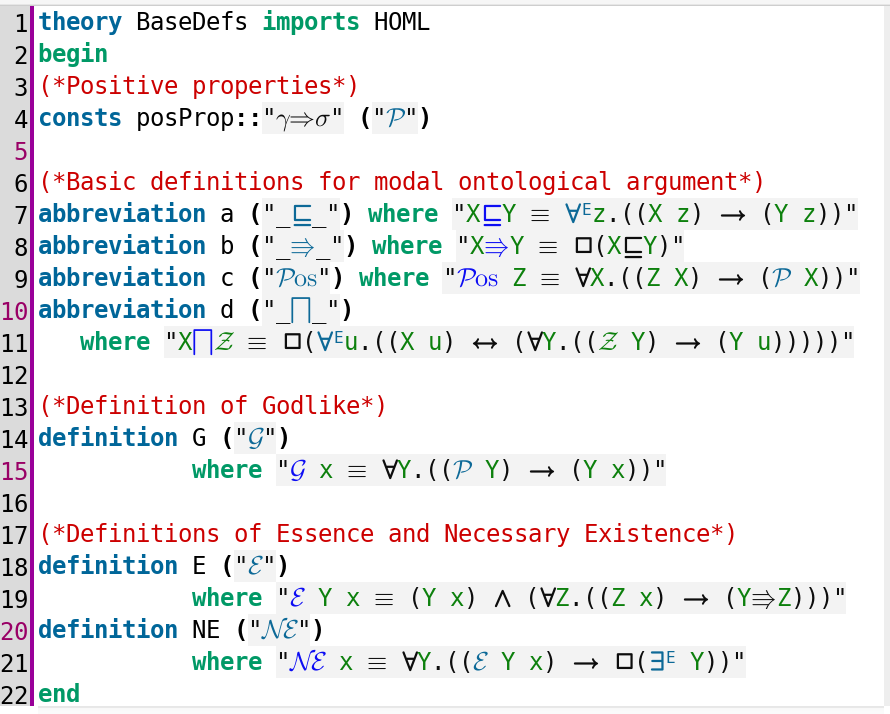}
\caption{Definitions for all variants discussed in the remainder.\label{fig:BaseDefs}}
\end{figure} 
In line 4 the uninterpreted constant symbol $\mathcal{P}$, for
``positive properties'', is declared; it has type
$\gamma{\Rightarrow}\sigma$. $\mathcal{P}$ thus denotes an intensional, world-depended
concept.  In lines 7--11 abbreviations for the previously discussed
relations and predicates $\boldsymbol{\sqsubseteq}$,
$\boldsymbol{\Rrightarrow}$, $\boldsymbol{\textstyle\bigsqcap}$ and 
$\mathcal{P}os$ are introduced. In lines 14--15, Gödel's notion
of ``being Godlike'' ($\mathcal{G}$) is defined, and in lines 18--21
the previously discussed definitions
for Essence ($\mathcal{E}$) and Necessary Existence ($\mathcal{NE}$) are given.

The full formalization of Scott's variant of Gödel's argument is
presented as theory \textsf{\small ScottVariant} in
Fig.~\ref{fig:ScottVariant}. This theory imports and relies on the
previously introduced notions from theory files \textsf{\small HOML}, \textsf{\small
  MFilter} and \textsf{\small BaseDefs}. 

The premises of Gödel's argument, as already discussed earlier,
are stated in lines 4--10. In line 12 a semantical counterpart
\textsf{\small B'}  (symmetry of the accessibility relation \textbf{r}
associated with the $\boldsymbol{\Box}$-operator) of the
\textsf{\small B} axiom is proved.

\begin{figure}[tp!] \centering
\includegraphics[width=1\columnwidth,height=.96\textheight]{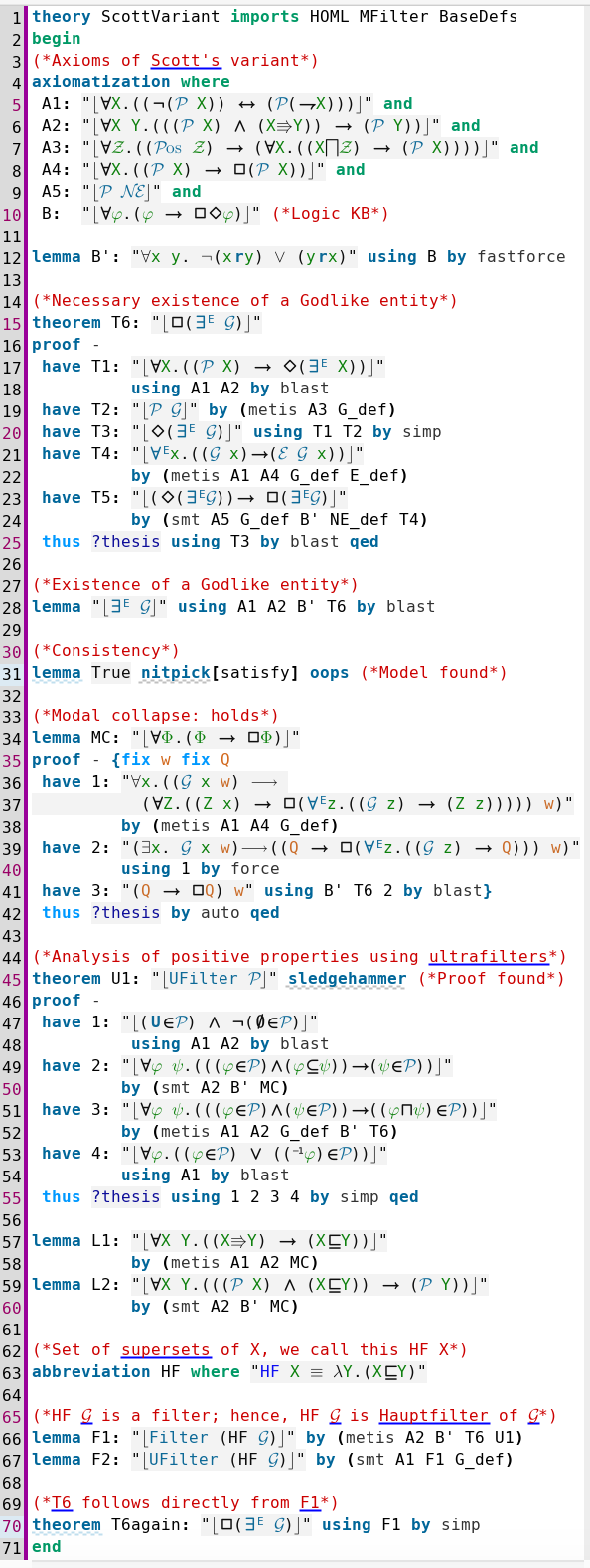}
\caption{Gödel's modal ontological argument; Scott's
  variant.\label{fig:ScottVariant}}
\end{figure} 

% \footnote{
% Remark: whether we use actualist or
% possibilist quantifiers for individuals, e.g., in the defns.~of $\boldsymbol\sqsubseteq$ or
% \textsf{\small T4} turned out irrelevant in this paper, and we consistently use
% actualist quantifiers in the remainder.
% }

An abstract level ``proof net'' for theorem \textsf{\small T6}, the
necessary existence of a Godlike entity, is presented in lines
15--25. Following the literature, the proof goes as follows: From
\textsf{\small A1} and \textsf{\small A2} infer \textsf{\small T1}:
positive properties are possibly exemplified. From \textsf{\small A3}
and the defn.~of $\mathcal{G}$ obtain \textsf{\small T2}: being
Godlike is a positive property (Scott actually directly postulated
\textsf{\small T2}). 
Using \textsf{\small T1} and \textsf{\small T2}
show \textsf{\small T3}: possibly a Godlike entity exists. Next, use
\textsf{\small A1}, \textsf{\small A4}, the defns.~of $\mathcal{G}$
and $\mathcal{E}$ to infer \textsf{\small T4}: being Godlike is an
essential property of any Godlike entity. From this, \textsf{\small
  A5}, \textsf{\small B'} and the defns.~of $\mathcal{G}$, and
$\mathcal{NE}$ have \textsf{\small T5}: the possible existence of a
Godlike entity implies its necessary existence. \textsf{\small T5} and
\textsf{\small T3} then imply \textsf{\small T6}.

The five subproofs and their dependencies have been automatically
proved using state-of-the-art ATP system integrated with
Isabelle/HOL via its \textit{sledgehammer} tool; sledgehammer then
identified and returned the abstract level proof justifications as
displayed here, e.g.~``using \textsf{\small T1} \textsf{\small T2} by
\textsf{\small simp}''. The mentioned proof engines/tactics
\textsf{\small blast}, \textsf{\small metis}, and \textsf{\small simp} 
are trustworthy components of Isabelle/HOL's, since they internally
reconstruct and check each (sub-)proof in the proof assistants small
and trusted proof kernel. The \textsf{\small smt} method, which relies
on an external satisfiability modulo solver (CVC4 in our case), is
less trusted, but we nevertheless use it here since it was the only
Isabelle/HOL method that was able to close this subproof in a single
step (we want to avoid displaying longer interactive proofs due to
space restrictions).  Using
the defns.~from Sect.~\ref{sec:HOML}, one can generally
reconstruct and verify all presented proofs with pen and paper directly
in meta-logic HOL. Moreover, reconstruction of modal logic proofs from
such proof nets within direct proof calculi for quantified modal
logics, cf.~Kanckos and Woltzenlogel-P.~\shortcite{Kanckos2017VariantsOG} or
Fitting~\shortcite{fitting02:_types_tableaus_god} is also possible.
  % We leave this as an
  % exercise to reader (this analogously applies to all argument variants
  % as discussed in the remainder).

The presented theory is consistent, which is confirmed in line 31 by
model finder \textit{nitpick}; \textit{nitpick} reports a
model consisting of one world and one Godlike entity.

Validity of modal collapse (\textsf{\small MC}) is confirmed in lines 34--42; a proof
net displaying the proofs main idea is shown.

Most relevant for this paper is that the ATP systems were able to
quickly prove that Gödel's notion of positive properties $\mathcal{P}$
constitutes a modal ultrafilter, cf.~lines 45--55. This was key to the
idea of taking the modal ultrafilter property of $\mathcal{P}$ as an
axiom \textsf{\small U1} of the theory; see the next section.

In lines 57--60 some further relevant lemmata are proved.  And in line
62--70 we hint at a much simpler, alternative proof argument: Take the set
$\text{HF}\,\mathcal{G}$ of all supersets of $\mathcal{G}$; it follows
from Gödel's theory that this set is a modal filter (line 66)
resp.~modal ultrafilter (line 67), i.e.,
$\text{HF}\,\mathcal{G}$ is the modal Hauptfilter of $\mathcal{G}$.
The necessary existence of a Godlike entity now becomes a simple corollary of this
result (see line 70, where \textsf{\small T6Again} is proved
exclusively from \textsf{\small F1}).
 In Sect.~\ref{sec:Hauptfiltervariant} we will later present an argument variant that is based on this observation.%
%\footnote{It was actually xxx at the University of yyy who pointed me to this proof argument alternative in an email communication.}%
\footnote{Manfred Droste from the University of Leipzig pointed me to this proof argument alternative in an email communication.}

% \begin{figure}[tp] \centering
% \includegraphics[width=1\columnwidth]{ScottVariant2}
% \caption{Scott's variant (cont'd
%   from Fig. \ref{fig:ScottVariant1}). \label{fig:ScottVariant2}}
% \end{figure}

\section{Ultrafilter Variant} \label{sec:newargument} 
\begin{figure}[t] \centering
\includegraphics[width=1\columnwidth]{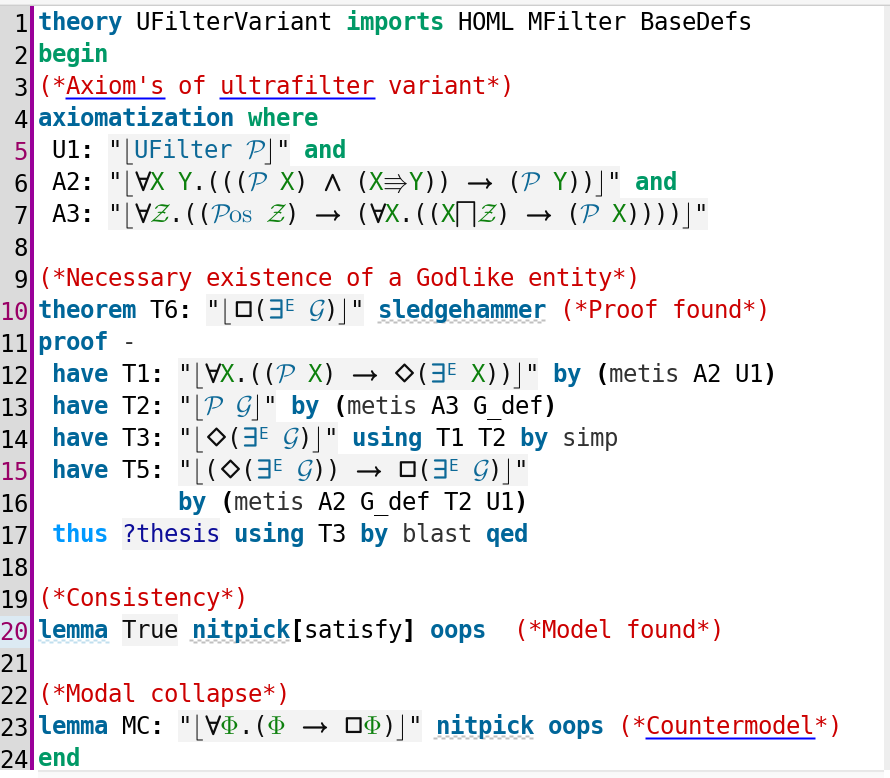}
\caption{Ultrafilter variant. \label{fig:UFilterVariant}}
\end{figure}
Taking \textsf{\small U1} as an axiom for Gödel's theory in fact
leads to a significant simplification of the modal ontological
argument; this is shown in lines 10--17 of the theory file  \textsf{\small UFilterVariant} 
in Fig.~\ref{fig:UFilterVariant}: not only Gödel's axiom \textsf{\small A1}
can be dropped, but also axioms \textsf{\small A4} and \textsf{\small A5}, together with
defns.~$\mathcal{E}$ and $\mathcal{NE}$. Even logic \textbf{KB} can be given up, since
\textbf{K} is now sufficient for verifying the proof argument.

The proof is similar to before: Use \textsf{\small U1} and
\textsf{\small A2} to infer \textsf{\small T1} (positive properties are
possibly exemplified). From \textsf{\small A3} and defn.~of
$\mathcal{G}$ have \textsf{\small T2} (being Godlike is a positive
property). \textsf{\small T1} and \textsf{\small T2} imply
\textsf{\small T3} (a Godlike entity possibly exists). From
\textsf{\small U1}, \textsf{\small A2}, \textsf{\small T2} and the
defn.~of $\mathcal{G}$ have \textsf{\small T5} (possible existence of
a Godlike entity implies its necessary existence).  Use
\textsf{\small T5} and \textsf{\small T3} to conclude \textsf{\small
  T6} (necessary existence of a Godlike entity).

% \begin{figure}[!t!p] \centering
% 		\begin{tikzcd}[row sep=1em,column sep=1em]
% 		{}  \& A50 \arrow{d}{@} \arrow{dr}{@} \& {} \& A51 \arrow{dl}{@} \& {} \\[1em]
% 		A48 \arrow{dr} \& A47 \arrow{d} \& A49 \arrow{dd} \& A45 \arrow{ddl} \& A46 \arrow{ddll} \\[1em]
% 		{} \& * \arrow{dr} \& {} \& {} \& {} \\[1em]
% 		{} \& {}\& A22 \& {} \& {}
% 		\end{tikzcd}
% 	\caption{} \label{fig:Argument}
% \end{figure}%

Consistency of the theory is confirmed in line 20; again a model with one
world and one Godlike entity is reported.

Most interestingly, modal collapse \textsf{\small MC} now has a simple
countermodel as \textit{nitpick} informs us in line 23. This
countermodel consists of a single entity $e_1$ and two worlds $i_1$
and $i_2$ with accessibility relation 
$\boldsymbol{\textsf{r}} = \{ \langle i_1, i_1\rangle, \langle i_2,
i_1\rangle, \langle i_2, i_2\rangle\}$.
% $i_j\, \boldsymbol{\textsf{r}}\, i_k$ for
% $j,k\in\{1,2\}$.
Trivially, formula $\Phi$ is such that $\Phi$ holds in $i_2$ but
not in $i_1$, which invalidates \textsf{\small MC} at world $i_2$. $e_1$ is the Godlike
entity in both worlds, i.e., $\mathcal{G}$ is the property that holds for $e_1$
in $i_1$ and $i_2$, which we may denote as
$\lambda e. \lambda w. e{=}e_1 \wedge (w{=}i_1 \vee w{=}i_2)$.
Using tuple notation we may write $\mathcal{G}=\{ \langle e_1, i_i \rangle, \langle e_1,
i_2 \rangle \}$.

Remember that $\mathcal{P}$, which is of type
$\gamma{\Rightarrow}\sigma$,  is
an intensional, world-depended concept. In our countermodel for
\textsf{\small MC}  in line 23 the
extension of $\mathcal{P}$ for world $i_1$ has the above $\mathcal{G}$ and
$\lambda e. \lambda w. e{=}e_1 \wedge w {=} i_1$ as its elements, while
in world $i_2$ we have $\mathcal{G}$ and
$\lambda e. \lambda w. e{=}e_1 \wedge w{=}i_2$. Using 
tuple notation we may note
$\mathcal{P}$ as \\[1em]
% %
% \hspace*{1cm} $\{ \langle G, i_1\rangle, \langle (\lambda e. \lambda w. e{=}e_1
% \wedge w {=} i_1), i_1\rangle,$ \\
% \hspace*{1cm} $\phantom{\{} \langle G, i_2\rangle, \langle (\lambda e. \lambda
% w. e{=}e_1 \wedge w {=} i_2) , i_2\rangle\}$ \\[1em]
% %
% or alternatively as \\[1em]
%
\hspace*{1cm} $\{ \langle \{ \langle e_1, i_1\rangle, \langle e_1, i_2\rangle \}, i_i \rangle,
\langle \{ \langle e_1, i_1\rangle  \}, i_i \rangle, $ \\
\hspace*{1cm} $\phantom{\{} \langle \{ \langle e_1, i_1\rangle, \langle e_1, i_2\rangle \}, i_2 \rangle,
\langle \{ \langle e_1, i_2\rangle  \}, i_2\rangle\}$ \\[1em]
%
% In the isolated context of world $i_1$ both
% $ \{ \langle e_1, i_1\rangle, \langle e_1, i_2\rangle \} $ and
% $ \{ \langle e_1, i_1\rangle \}$ are reducing to just $\{ e_1 \}$, and so are
% $ \{ \langle e_1, i_1\rangle, \langle e_1, i_2\rangle \} $ and
% $ \{ \langle e_1, i_2\rangle \}$ in world $i_2$. 
In order to verify that $\mathcal{P}$ is a modal ultrafilter we have
to check whether the respective modal ultrafilter conditions are
satisfied in both worlds. 
$\textbf{U}\boldsymbol\in\mathcal{P}$ in $i_1$ and also in
$i_2$, since both
$\langle \{ \langle e_1, i_1\rangle, \langle e_1, i_2\rangle  \}, i_1\rangle$ and
$\langle \{ \langle e_1, i_1\rangle, \langle e_1, i_2\rangle  \},
i_2\rangle$ are in
$\mathcal{P}$;  
$\boldsymbol\emptyset\boldsymbol\not\boldsymbol\in\mathcal{P}$ in $i_1$ and also in
$i_2$, since
both
$\langle \{ \}, i_1\rangle$ and
$\langle \{ \}, i_2\rangle$ are not in
$\mathcal{P}$. It is also easy to verify that $\mathcal{P}$ is closed under supersets and
intersection in both worlds.

Note that in our countermodel for \textsf{\small MC}, also Gödel's
axiom $\textsf{\small A4}$ is invalidated. Consider
$X=\lambda e. \lambda w. e{=}e_1 \wedge w{=}i_2$, i.e., $X$ is true
for $e_1$ in $i_2$, but false for $e_1$ in $i_1$. We have
$\mathcal{P}\, X$ in $i_2$, but we do not have
$\Box (\mathcal{P}\, X)$ in $i_1$, since $\mathcal{P}\, X$ does not hold
in $i_1$, which is reachable in \textbf{r} from $i_2$.

\textit{nitpick} is capable of computing all \textit{partial}
modal ultrafilters as part of its countermodel exploration: out of 512
candidates, \textit{nitpick} identifies 32 structures of form
$\langle F, i\rangle$, for $i\in\{1,2\}$, in which $F$ satisfies the ultrafilter
conditions in the specified world $i$. An example for such an $\langle
F, i\rangle$ is \\[1em]
\hspace*{1cm} $\langle \{ \langle \{ \langle e_1, i_1\rangle  \}, i_i \rangle, \langle \{ \}, i_i \rangle$, \\
\hspace*{1cm} $\phantom{\langle \{} \langle \{ \langle e_1, i_1\rangle, \langle e_1, i_2\rangle \}, i_2 \rangle,
\langle \{ \langle e_1, i_2\rangle  \}, i_2\rangle\}, i_2\rangle$
\\[1em]
$F$ is not a
proper modal ultrafilter, 
%since $\langle F, i_1\rangle$ does not hold.
since $F$ fails to be an ultrafilter in world $i_1$.

      % ((λx. _)
      %   (((λx. _)((e1, i1) := True, (e1, i2) := True), i1) := True,
      %    ((λx. _)((e1, i1) := True, (e1, i2) := True), i2) := True,
      %    ((λx. _)((e1, i1) := True, (e1, i2) := False), i1) := True,
      %    ((λx. _)((e1, i1) := True, (e1, i2) := False), i2) := False,
      %    ((λx. _)((e1, i1) := False, (e1, i2) := True), i1) := False,
      %    ((λx. _)((e1, i1) := False, (e1, i2) := True), i2) := True,
      %    ((λx. _)((e1, i1) := False, (e1, i2) := False), i1) := False,
      %    ((λx. _)((e1, i1) := False, (e1, i2) := False), i2) := False),
      %   i1) :=
      %    TrueBLABLA,
      %  ((λx. _)
      %   (((λx. _)((e1, i1) := True, (e1, i2) := True), i1) := True,
      %    ((λx. _)((e1, i1) := True, (e1, i2) := True), i2) := True,
      %    ((λx. _)((e1, i1) := True, (e1, i2) := False), i1) := True,
      %    ((λx. _)((e1, i1) := True, (e1, i2) := False), i2) := False,
      %    ((λx. _)((e1, i1) := False, (e1, i2) := True), i1) := False,
      %    ((λx. _)((e1, i1) := False, (e1, i2) := True), i2) := True,
      %    ((λx. _)((e1, i1) := False, (e1, i2) := False), i1) := False,
      %    ((λx. _)((e1, i1) := False, (e1, i2) := False), i2) := False),
      %   i2) :=
      %    TrueBLABLA,

\section{Simplified Variant} 
\begin{figure}[t] \centering
\includegraphics[width=1\columnwidth]{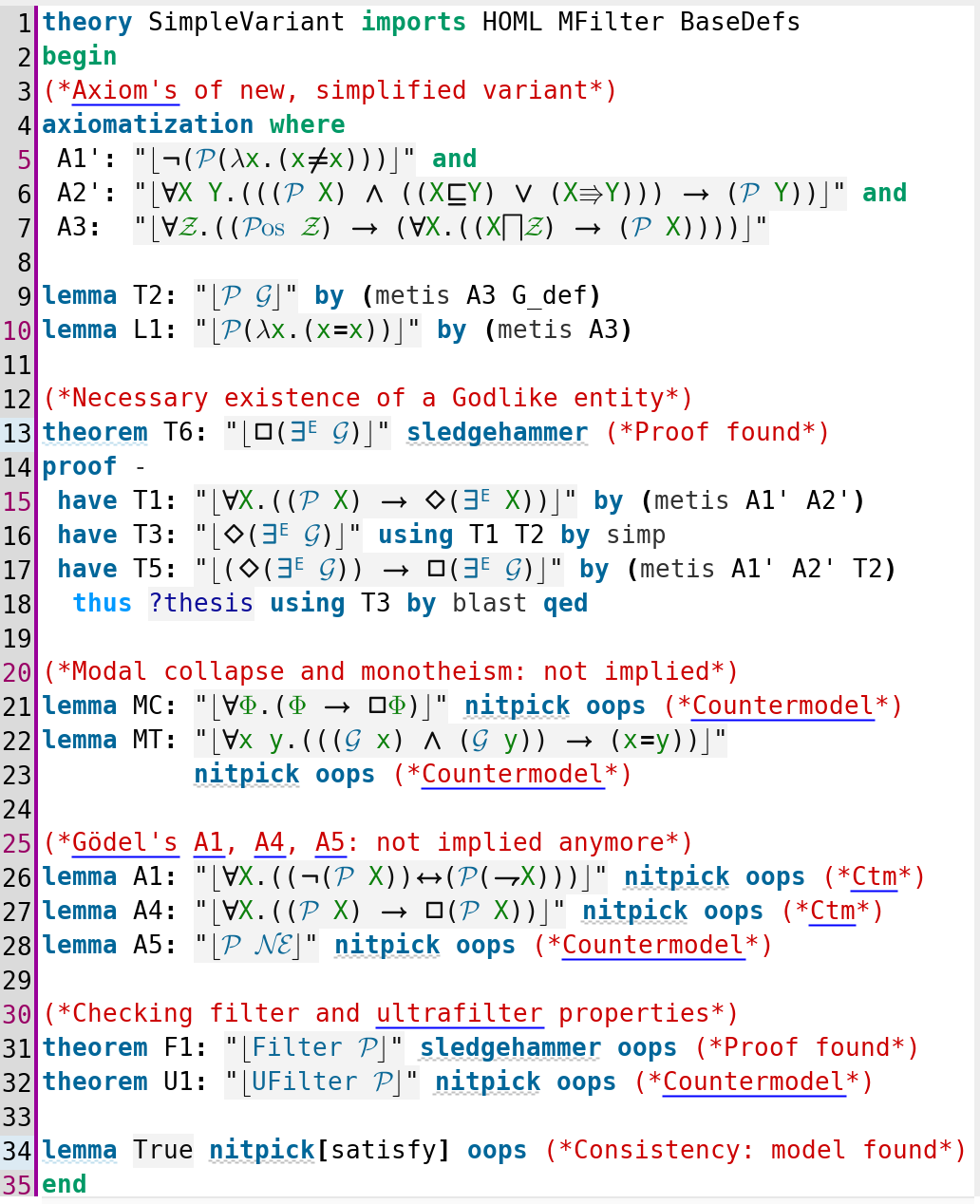}
\caption{Simplified variant. \label{fig:SimpleVariant}}
\end{figure}
\label{sec:finalargument} What modal ultrafilters properties of
$\mathcal{P}$ are actually needed to support \textsf{\small T6}? Which
ones can be dropped?  Experiments with our framework, as displayed in
theory file \textsf{\small SimpleVariant} in
Fig.~\ref{fig:SimpleVariant}, confirm that only the filter conditions from
Sect.~\ref{sec:ultrafilter} must be upheld for $\mathcal{P}$; 
maximality can be dropped. However, it is possible to merge filter condition~3
(closed under supersets) for $\mathcal{P}$ with Gödel's \textsf{\small
  A2} into axiom \textsf{\small A2'} as shown in line 6 of
Fig.~\ref{fig:SimpleVariant}.  Moreover, instead of requiring that the
empty set $\boldsymbol\emptyset = \lambda x. \boldsymbol\bot$ must not
be a positive property, we postulate that self-difference
$\lambda x. x{\boldsymbol{\not=}} x$ is not (line 5); note that
self-difference is extensionally equal to
$\boldsymbol\emptyset$. Self-identity and self-difference have been
used frequently in the history of the ontological argument, which is
part of the motivation for this switch. As intended, filter
condition~4 is now implied by the theory (see theorem \textsf{\small
  F1} proved in line 31), as well as positiveness of self-identity
(line 10). The essential idea of the theory \textsf{\small
  SimpleVariant} in Fig.~\ref{fig:SimpleVariant} is to show that it
actually suffices, in combination with \textsf{\small A3}, to
postulate that $\mathcal{P}$ is a modal filter, and this is what our
simplified axioms do.

From the definition of $\mathcal{G}$ and the axioms
\textsf{\small A1'}, \textsf{\small A2'} and \textsf{\small A3} (lines
5--7) theorem \textsf{\small T6} immediately follows: in line 13
several theorem provers integrated with \textit{sledgehammer} quickly report a
proof ($\leq$ 1sec). Moreover, a more detailed and more intuitive ``proof net'' is presented in lines
14--18; the proof argument is analogous to what has been discussed before.

In lines 21--23, 
countermodels for modal collapse \textsf{\small MC} (similar to the one discussed
before) and 
for monotheism \textsf{\small MT} are reported. 

Further questions are answered experimentally (lines 26--38): neither
$\textsf{\small A1}$, nor $\textsf{\small A4}$ or $\textsf{\small A5}$,
of the premises we dropped from Gödel's theory are implied anymore,
all have countermodels.
In lines 31--32 we see that
$\mathcal{P}$ is still a filter, but not an ultrafilter.
Since some of these axioms, e.g.~Gödel's strong $\textsf{\small A1}$, have been
discussed controversially in the history of Gödel's argument, and since
\textsf{\small MC} and \textsf{\small MT} are independent, we
have arrived at a philosophically and theologically potentially
relevant simplification of Gödel's work.

\section{Further Simplified Variants} \label{sec:further}
\subsubsection{Postulating \textsf{T2} instead of \textsf{A3}}
Instead of working with third-order axiom \textsf{\small A3} to infer
\textsf{\small T2} as in theory \textsf{\small SimpleVariant}, we directly postulate \textsf{\small T2} 
as an axiom in theory \textsf{\small SimpleVariantPG}; cf. the new axiom \textsf{\small T2}  in line 7 of  in Fig.~\ref{fig:SimpleVariantPG}.

Theorem \textsf{\small T6} can be proved essentially as before (lines 10--15), and 
% : from
% \textsf{\small A1'} and \textsf{\small A2'} get \textsf{\small T1} (positive properties are possibly exemplified); from \textsf{\small T1}
% and \textsf{\small T2} infer \textsf{\small T3} (possible existence
% of a Godlike entity); from \textsf{\small A1'}, \textsf{\small A2'}
% and \textsf{\small T2} obtain \textsf{\small T5} (the possible
% existence of a Godlike entity implies its necessary existence);
% combining \textsf{\small T3} and \textsf{\small T5} have
% \textsf{\small T6} (the necessary existence of a Godlike entity).
\textsf{\small MC} and \textsf{\small MT} still have countermodels (lines
20--22).

\begin{figure}[t] \centering
\includegraphics[width=1\columnwidth]{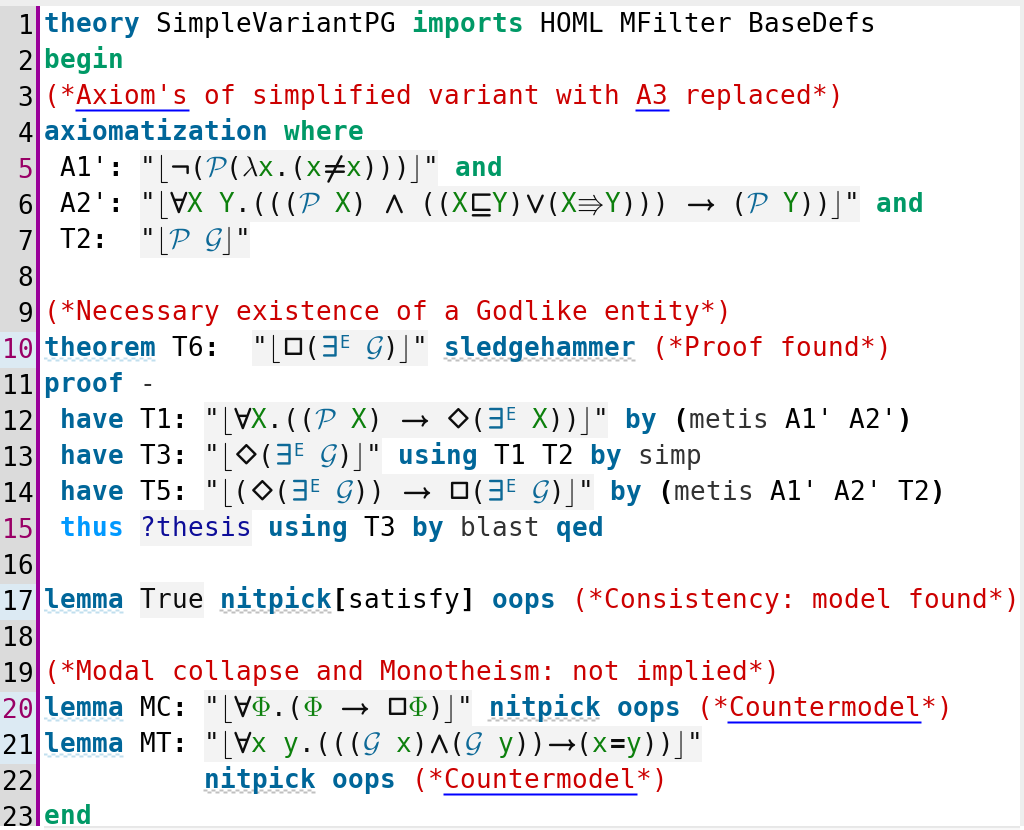}
\caption{Simplified variant with axiom \textsf{\scriptsize T2}. \label{fig:SimpleVariantPG}}
\end{figure}

\subsubsection{Simple Entailment in Axiom
  \textsf{A2'}} \label{sec:SimpleVariantSE}
\begin{figure}[t!] \centering
\includegraphics[width=1\columnwidth]{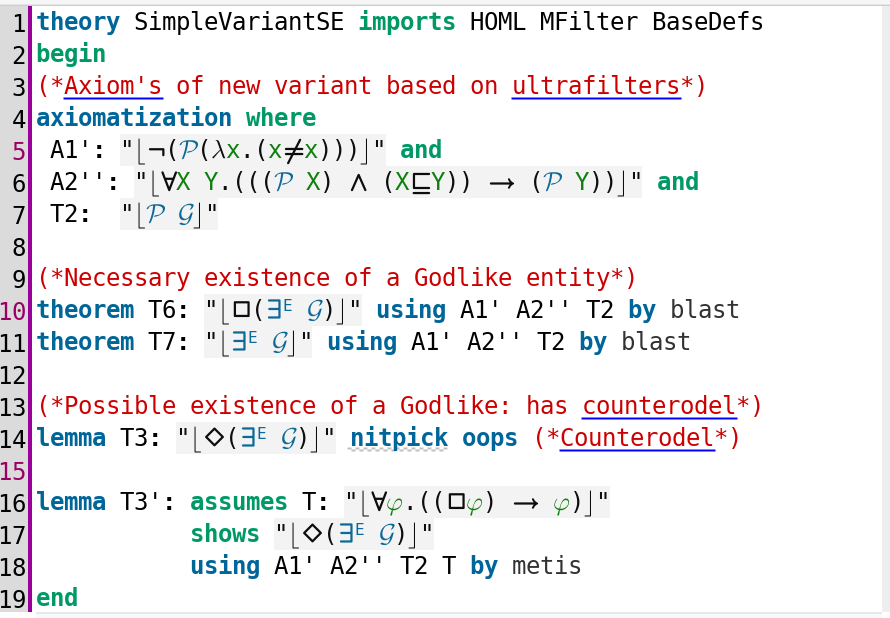}
\caption{Simplified variant with simple entailment in logic \textbf{K}. \label{fig:SimpleVariantSE}}
\end{figure}
Instead of using a disjunction of simple entailment and necessary
entailment in axiom \textsf{\small A2'} we may in fact only require
simple entailment in \textsf{\small A2'}; see axiom \textsf{\small
  A2''} (line 6) of the theory file \textsf{\small SimpleVariantSE}
displayed in
Fig.~\ref{fig:SimpleVariantSE}. 
Proofs for
\textsf{\small T6}, the necessary existence of a Godlike entity, and
also \textsf{\small T7},  
existence of a Godlike entity, can still
be quickly found   (lines 10--11). 

However, after replacing \textsf{\small A2'}
by \textsf{\small A2''}, \textsf{\small T3} (the possible existence of
a Godlike entity) is no longer implied; see line~14. As nitpick
informs us, \textsf{\small
  T3} now has an undesired countermodel
consisting of one single world that is not connected to itself.
By assuming modal axiom \textsf{\small T} (what is necessary
true is true in the given world) this countermodel can be
eliminated so that \textsf{\small T3} is implied as desired (lines 16--18).

\subsubsection{Simple Entailment in Logic \textbf{T}}
The above discussion motivates a further alternative of the
simplified modal ontological argument; see theory file \textsf{\small
  SimpleVariantSEinT} in 
Fig.~\ref{fig:SimpleVariantSEinT}.
This argument is assuming 
modal logic \textbf{T} (which comes with axiom \textsf{\small T} as
discussed above), and, as before, it postulates
axioms \textsf{\small A1'} \textsf{\small A2''} and \textsf{\small T2}
(lines 5--7):
\begin{description}
%\item[\textsf{\small G}]  A Godlike entity to possess all positive properties.
\item[\textsf{\small A1'}] Self-difference is not a positive property.
\item[\textsf{\small A2''}] A property entailed by a positive property is positive. 
\item[\textsf{\small T2}] Being Godlike is a positive property.
\end{description}

\begin{figure}[t] \centering
\includegraphics[width=1\columnwidth]{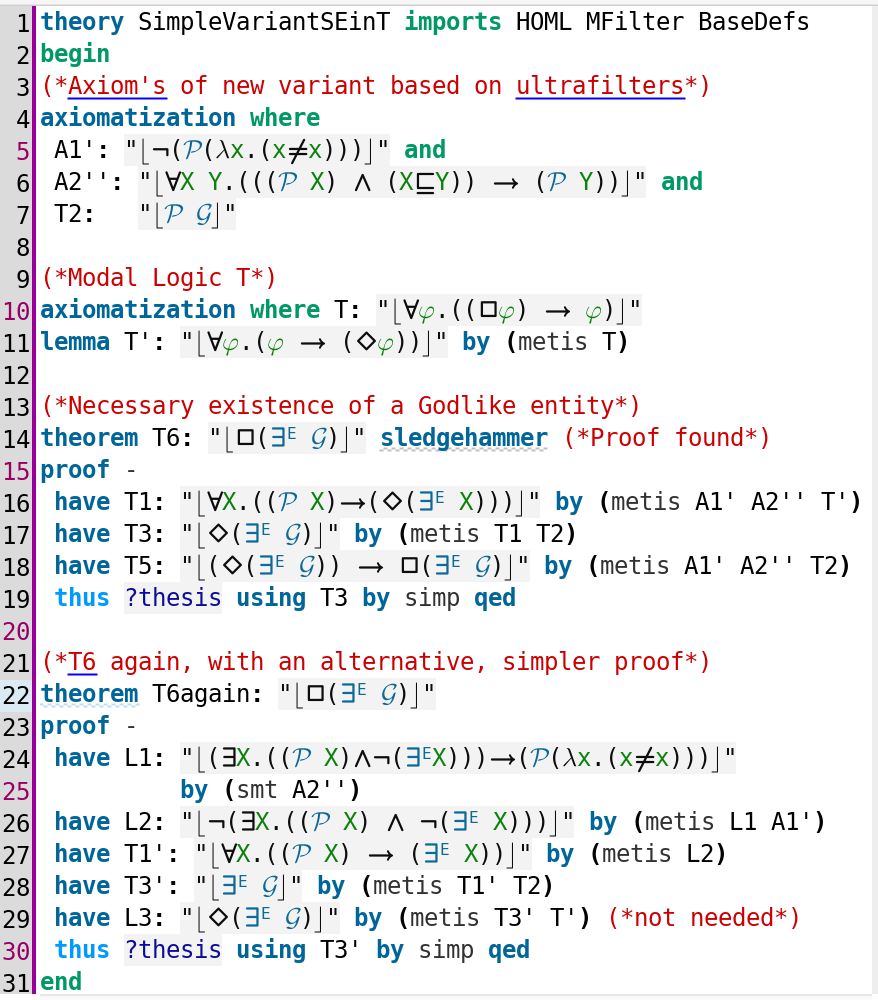}
\caption{Simplified variant with simple entailment in logic \textbf{T}. \label{fig:SimpleVariantSEinT}}
\end{figure}

% \begin{figure}[t] \centering
% \includegraphics[width=1\columnwidth]{SimpleVariantSEinTpossibilist.png}
% \caption{Variant with simple entailment in modal logic \textbf{T} with
%   possibilist individual quantifiers. \label{fig:SimpleVariantSEinTpossibilist}}
% \end{figure}
One possible proof argument for \textsf{\small T6}  is as before; see lines
15--19. However, there is also a much simpler proof, see
lines 23--30, which we explain in more detail (this
simple proof is applicable to 
previous variants as well; it is also key to proving
\textsf{\small T5} from \textsf{\small A1'}, \textsf{\small A2'/A2''} and
\textsf{\small T2/A3} in previous variants, including the simplified
variant in Fig.~\ref{fig:SimpleVariant}):\footnote{The existence of such
  an unintended derivation was already hinted at by Fuenmayor
  and Benzmüller \shortcite[Footn.~11]{C65}.}
\begin{description}
\item[L1] {The existence of a non-exemplified 
positive property implies that self-difference is a positive
property}---\textit{This follows from axiom \textsf{\small A2''}.}
% \item[L2] \textit{Any true proposition is possibly true.} Follows from
%   contraposition of axiom \textsf{\small T}.
% \item[L3]  \textit{If there exists a 
% positive property whose exemplification is impossible then self-difference is positive
% property.} This follows from \textsf{\small L1} and \textsf{\small
% L2}.
\item[L2] {There is no non-exemplified positive property}---\textit{From \textsf{\small L1} and axiom \textsf{\small A1'}.}
\item[T1'] {Positive properties are exemplified}---\textit{Equivalent to \textsf{\small L2}.}
\item[T3'] {A Godlike entity exists}---\textit{From\textsf{\small T1'} and axiom \textsf{\small T2}.}
\item[L3] {A Godlike entity possibly exists}---\textit{From
  \textsf{\small T3'} and \textsf{\small T'} (contrapositive of
  axiom T); note that \textsf{\small L3} is not needed to obtain
  \textsf{\small T6} in the next step; generally, axiom \textsf{\small
    T} (resp.~its contrapositive \textsf{\small T'}) is only needed
  for deriving \textsf{\small T3/L3}.}
\item[T6] {A Godlike entity necessarily exists}---\textit{From
  \textsf{\small T3'} by necessitation.}
\end{description}

% We leave it as a technical exercise to provide a tableau proof
% \cite{fitting02:_types_tableaus_god} or natural
% deduction proof \cite{Kanckos2017VariantsOG} for the proof arguments as
% just discussed.

% Finally in
% Fig.~\ref{fig:SimpleVariantSEinTpossibilist} we show
% that 

% The presented proof argument still works if we replace all actualist individual
% quantifiers $\boldsymbol{\exists}^E$ and $\boldsymbol{\forall}^E$ by
% respective possibilist individual
% quantifiers $\boldsymbol{\exists}$ and $\boldsymbol{\forall}$ (not
% shown here). 

\subsubsection{Hauptfiltervariant} \label{sec:Hauptfiltervariant}
\begin{figure}[t] \centering
\includegraphics[width=1\columnwidth]{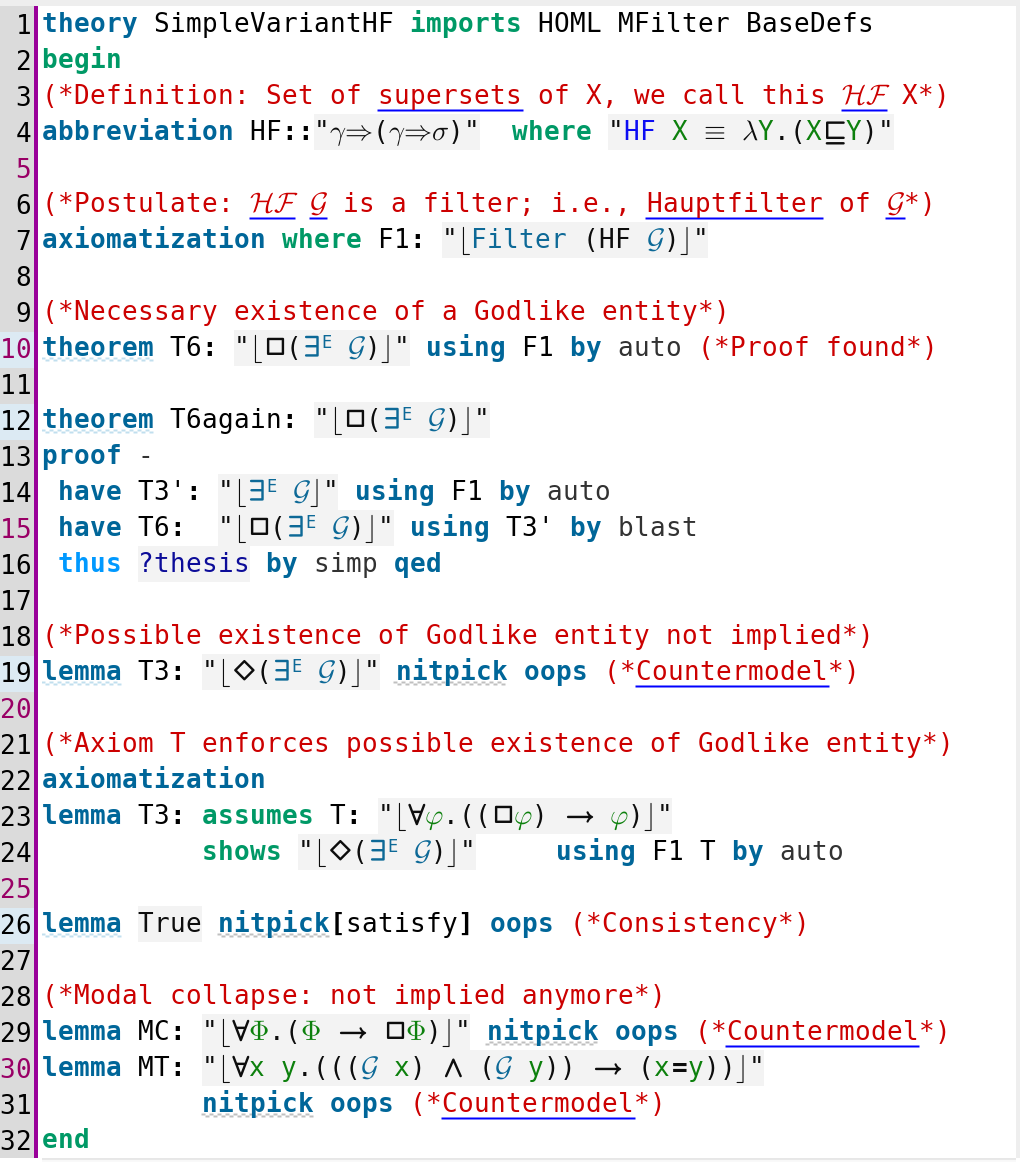}
\caption{Hauptfiltervariant. \label{fig:SimpleVariantHF}}
\end{figure}
Another drastically simplified variant of the modal ontological
argument is related 
to the observations discussed earlier at the end of
Sec.~\ref{sec:goedelargument}. There it has been shown that the set
$\text{HF}\,\mathcal{G}$, consisting of all supersets of
$\mathcal{G}$, is a modal filter (so that $\text{HF}\,\mathcal{G}$ is the 
Hauptfilter of $\mathcal{G}$); this then directly implies the
necessary existence of a Godlike entity.  A new variant based on this
observation is presented in theory file \textsf{\small
  SimpleVariantHF} in Fig.~\ref{fig:SimpleVariantHF}.  Here the set of
all supersets $\text{HF}\,\mathcal{G}$ of  property
$\mathcal{G}$ (being Godlike) is postulated to be a modal filter (axiom
\textsf{\small F1} in line 7).\footnote{In fact, the only 
essential requirement that is enforced here is that
$\boldsymbol{\emptyset}$ is not in $\text{HF}\,\mathcal{G}$, hence
$\mathcal{G}$ cannot be identical to $\boldsymbol{\emptyset}$.}
% ; in fact, this essentially boils down to
% postulating that $\boldsymbol{\emptyset}$ is not in
% $\mathcal{HF}\, \mathcal{G}$.
The existence and necessary existence
of a Godlike entity then directly follows from \textsf{\small F1}, see
lines 10--16. And as already  observed  before, the possible existence
of a Godlike entity is independent, but can be enforced by postulating
modal axiom \textsf{\small T} (lines 19--24). Moreover, the theory
is consistent (line 26) and neither modal collapse nor monotheism are implied
(lines 29--30).

\section{Related Work} \label{sec:relwork}

Fitting \shortcite{fitting02:_types_tableaus_god} suggested to
carefully distinguish between intensions and extensions of positive
properties in the context of G\"odel's modal ontological argument, and, in order to do
so within a single framework, he introduced a sufficiently expressive
HOML enhanced with means for the explicit representation of
intensional terms and their extensions; cf.~also the intensional
operations used by Fuenmayor and Benzmüller \shortcite{C65,J52}
in their formal study of the works of 
Fitting and
Anderson \shortcite{Anderson90,AndersonGettings}. % and it introduces just a subset of
% the operators studied as needed by Fuenmayor
% and Benzmüller for their ultrafilter-based comparative assessment of
% the works by Fitting~\shortcite{fitting02:_types_tableaus_god},
% Anderson~\shortcite{Anderson90,AndersonGettings} and
% Scott~\shortcite{ScottNotes}.

The application of computational methods to philosophical problems was
initially limited to first-order theorem provers. Fitelson and Zalta
\shortcite{FitelsonZalta} used \textsc{Prover9} to find a proof of the theorems
about situation and world theory in \cite{zalta1993} and they found an
error in a theorem about Plato's Forms that was left as an exercise in
\cite{PelletierZalta}. Oppenheimer and Zalta
\shortcite{OppenheimerZalta2011} discovered, using Prover9, that one of
the three premises used in their reconstruction of Anselm's ontological
argument \cite{OppenheimerZalta1991} was sufficient to derive the
conclusion. The first-order conversion techniques that were developed
and applied in these works are outlined in some detail in related work
by Alama, Oppenheimer and Zalta \shortcite{AlamaZalta2015}.

More recent related work makes use of higher-order proof
assistants. Besides the already mentioned work of
Benzmüller and colleagues, this includes Rushby's~\shortcite{Rushby} study on Anselm's ontological argument in 
the PVS system and Blumson's \shortcite{Blumson} related study in
Isabelle/HOL.

The development of Gödel's ontological argument has recently been
addressed by Kanckos and Lethen
\shortcite{KanckosLethen19}.  They discovered previously unknown variants
of the argument in Gödel's Nachlass, whose relation to the presented simplified variants
should be further investigated in future work.\footnote{An email discussion
  (March, 2020) with Tim Lethen revealed the following: In particular
  version No.~2 of Gödel's argument as presented in
  \cite{KanckosLethen19} appears related, though not equivalent, to
 our simplified versions. Version No.~2 --which we have meanwhile
 formalized and verified in Isabelle/HOL
(cf.~Fig.~\ref{fig:No2Possibilist} in the appendix)
 -- avoids the
  notions of essence and necessary existence and associated
  definitions/axioms, just as our simplified versions do. Moreover,
  their findings also suggest that instead of axiom
  \textsf{\scriptsize A1'} we may just postulate that a non-positive
  property exists (and experiments confirm this claim;
  \textsf{\scriptsize A1'} is then implied). However, I prefer
  axiom \textsf{\scriptsize A1'}.}

Related work also includes Odifreddi's
\shortcite{odifreddi00:_ultraf_dictat_gods} discussion on ultrafilters,
dictators and God, which I was pointed to by a reviewer.

\section{Discussion}
The simplifications of Gödel's theory as presented are far reaching. In fact,
one may  ask \textit{``Is this really what 
Gödel had in mind?''}, or are there some technical
issues, such as the alternative proof in lines 22--30 from
Fig.~\ref{fig:SimpleVariantSEinT}, 
that have been ignored so far? Moreover, is the
definition of property entailment ($\sqsubseteq$) really adequate in
the context of the modal ontological argument, or shouldn't this definition  be replaced by 
concept containment, so that self-difference,
resp.~the empty property, would
no longer $\sqsubseteq$-entail any other property?\footnote{See also
  \cite{AlamaZalta2015} and \cite[Footn.~11]{C65}.}

Moreover, assuming that the simplified theory \textsf{\small
  SimpleVariant} from 
Fig.~\ref{fig:SimpleVariant} is indeed still  in line with what
Gödel had in mind, why not presenting the definitions and axioms using
an alternative wording, for example, as follows (where we replace
``positive property'' by
``rational/consistent property'' and ``Godlike entity'' by
``maximally-rational entity''):

\begin{description}
\item[$\mathcal{G}$] An entity x is maximally-rational ($\mathcal{G}$) iff it has all rational/consistent properties.
\item[\textsf{\small A1'}] Self-difference is not a rational/consistent
  property.
\item[\textsf{\small A2'}] A property entailed or necessarily entailed by a rational/consistent property is rational/consistent.
\item[\textsf{\small T2}] Maximal-rationality is a rational/consistent
  property.
\item[It follows:] A maximally-rational entity necessarily exists.
\end{description}

It would still be possible, but not mandatory, to understand a
maximally-rational being also as a Godlike being.

% Discuss relation to ``All properties are divine, or God exists'' \cite{bjoerdal18}

Independent of this discussion we expect the Isabelle/HOL theory files
we contributed to be useful for teaching quantified modal logics in
classroom, as previously demonstrated by Wisniewski et al.~\shortcite{C58}
in their awarded lecture course on ``Computational Metaphysics''. The
developed corpus of example problems is furthermore suited as a
benchmark for other ambitious knowledge representation and reasoning
projects in the KR community: Can the alternative approaches encode
such metaphysical argumants as well? How about their proof
automation capabilities and how about model finding? For example, description
logic or argumentation theory, due to their limited
expressivity, appear unsuited to support such ambitious applications, while the techniques
presented in this paper have meanwhile been used successfully for
the encoding and assessment of 
foundational theories in metaphysics~\cite{J50} and
mathematics~\cite{J40,C78}.
Moreover, our problem set constitutes an interesting benchmark for other
HOL automated theorem provers and should therefore be converted into
TPTP THF \cite{J22} representation and be used in theorem prover competitions.

% Moreover, related work has shown that the meta-logical approach
% presented in this paper scales for applications beyond metaphysics,
% including, for example, the study of foundational systems for
% mathematics~\cite{J40,C78}.

\section{Conclusion}
Gödel's modal ontological argument stands in prominent tradition of
western philosophy. It has its roots in the
Proslogion of Anselm of Canterbury and it has been picked up in
Descartes' Fifth Meditation and in the works of Leibniz,
which in turn inspired and informed the work of Gödel.

In this paper we have linked Gödel's theory to a suitably adapted
mathematical theory (modal filter and modal ultrafilter), and subsequently
we have developed simplified modal ontological
arguments which avoid some of Gödel's axioms and consequences,
including modal collapse, that have led to criticism in the past. At
the same time the offered simplifications
are very far reaching, eventually too far. Anyhow, the insights
that were presented in this paper appear
relevant to inform ongoing studies of the modal ontological argument
in theoretical philosophy and theology.

%Regarding the ongoing debate on the relevance and difference of
%subsymbolic and symbolic AI technologies we may comment:
%While data scientists apply subsymbolic AI techniques to obtain
%approximating and rather opaque models in their application domains of
%interest, we
We have in this paper applied modern symbolic AI techniques
to arrive at deep, explainable and verifiable models of the
metaphysical concepts we are interested in. In particular, we have
illustrated how state of the art theorem proving systems, in
combination with latest knowledge representation and reasoning
technology, can fruitfully be employed to explore and contribute deep
new knowledge to other disciplines. 

% Here this technology has been applied to develop a
% significantly simplified and, as we find, quite elegant new variant of
% the modal ontological argument for the existence of a supreme being,
% that links and explains a theological concept with known mathematical
% structures. In fact, we conjecture that the new modal argument, due to
% its conciseness and the avoidance of modal collapse, will trigger some
% interest in philosophy and theology.

% Here this technology has been applied to develop a
% significantly simplified and, as we find, most elegant new variant of the
% modal ontological argument, that links and explains a theological theory with known
% mathematical structures. 

\appendix
\section*{Acknowledgements} 
I am grateful to friends and colleagues who have supported this line
of research in the past and/or directly contributed with comments
and suggestions, including (in alphabetical order) C. Brown, M. Droste, 
D. Fuenmayor, T. Gleißner, D. Kirchner, S. Kova\v{c}, T. Lethen, X.  Parent, R. Rojas, D. Scott, A. Steen,
L. van der Torre, E. Weydert, D. Streit, B. Woltzenlogel-Paleo,
E. Zalta. Moreover, I am grateful to the reviewers of this paper for
useful comments and for the pointer to P. Odifreddi's paper.

%\section*{Acknowledgements}  Acknowledgements  will be included in the final document.

{\small
\bibliographystyle{kr}
%\bibliography{main}

}

%\end{document}
\newpage
\begin{appendix}
\onecolumn
%\section*{Appendix}
\begin{figure*}[h!] \centering
\includegraphics[width=.73\textwidth]{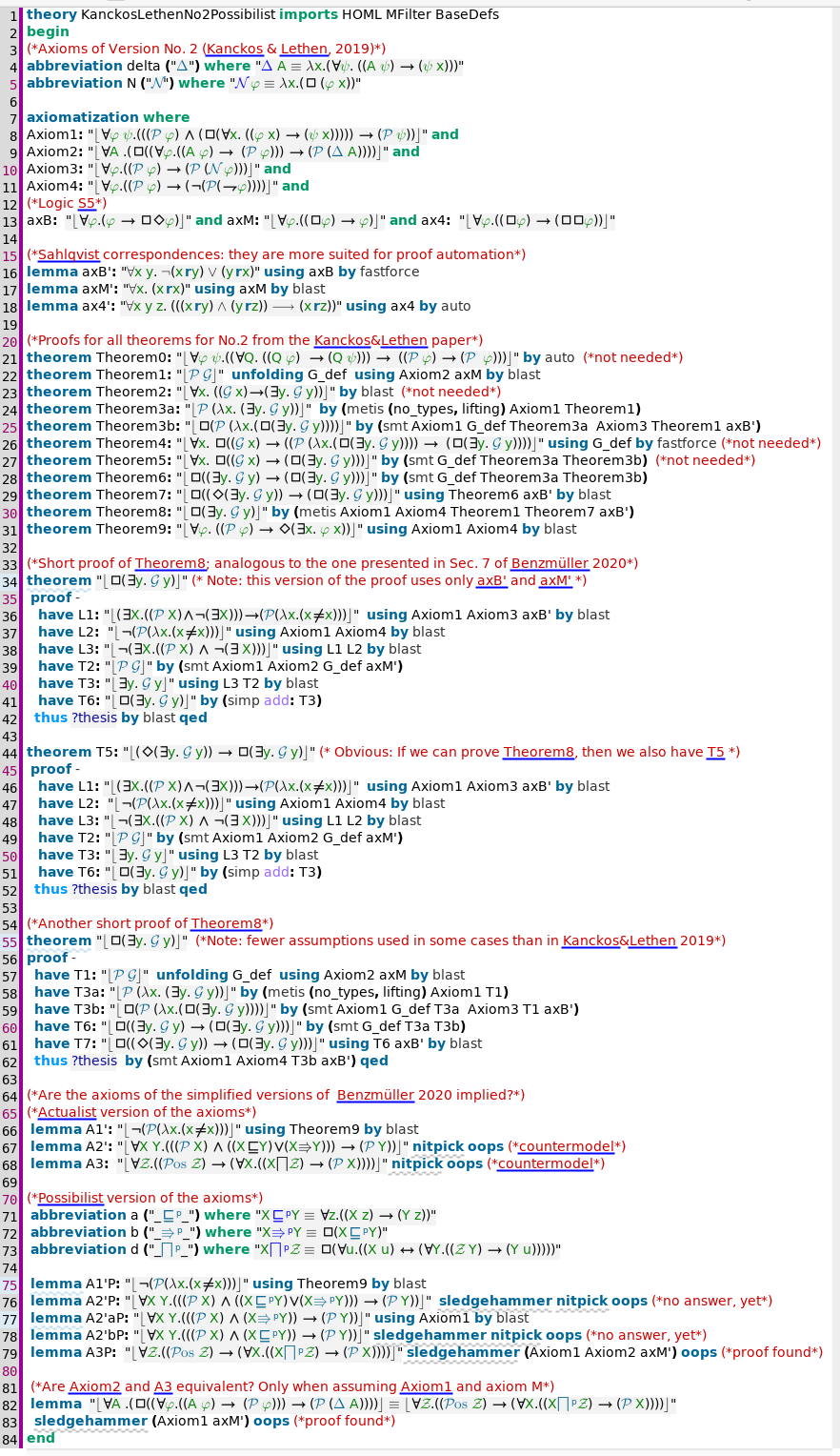}
\caption{Formal study of ``version No.2'' of Gödel's argument as reported
  by  Kanckos
  and Lethen (2019). \label{fig:No2Possibilist}}
\end{figure*}
\end{appendix}

\end{document}